\documentclass{llncs}

\usepackage[utf8]{inputenc}
\usepackage{amsmath, amssymb, amstext}
\usepackage{algorithm, algpseudocode, float}
\usepackage{theorem, thmtools, theoremref, thm-restate}
\usepackage{mathtools}
\usepackage{bm}
\usepackage{graphicx}
\usepackage{authblk}
\usepackage{hyperref}
\usepackage{textcomp}
\usepackage{enumitem}
\usepackage{color}

\newcommand{\np}{\ensuremath{\mbox{NP}}}
\newcommand{\st}{\ensuremath{\mbox{s.t.}}}

\def\OPT{\textsf{OPT}}
\def \cost{\textsf{cost}}

\DeclareMathOperator{\argmax}{arg\,max}

\makeatletter
\algnewcommand{\Indstatex}[1]{\Statex \hskip\ALG@thistlm #1}
 \makeatother

\newenvironment{proof*}[1]
  {\renewcommand{\proofname}{#1}
   \begin{proof}}
  {\end{proof}}

\begin{document}

\title{Approximating Stable Matchings with Ties of Bounded Size}

\author{Jochen Koenemann\inst{1} \and
Kanstantsin Pashkovich\inst{2} \and
Natig Tofigzade\inst{1}}

\authorrunning{J. Koenemann et al.}    

\institute{University of Waterloo, Waterloo ON N2L 3G1,  Canada\\ 
\email{\{jochen, natig.tofigzade\}@uwaterloo.ca}\\
\and
University of Ottawa, Ottawa ON K1N 6N5, Canada\\
\email{kpashkov@uottawa.ca}}

\maketitle

\begin{abstract} Finding a stable matching is one of the central
  problems in algorithmic game theory. If participants are allowed to
  have ties and incomplete lists, computing a stable matching of maximum cardinality is known to be \np-hard. In this paper we
  present a $(3L-2)/(2L-1)$-approximation algorithm for the stable
  matching problem with ties of size at most $L$ and incomplete preferences. Our result matches
  the known lower bound on the integrality gap for the associated LP formulation.  \end{abstract}
	
\section{Introduction}

In an instance of the classical stable matching problem we are given a
(complete) bipartite graph $G=(A \cup B, E)$ where, following standard
terminology, the nodes in $A$ will be referred to as {\em men}, and
the nodes in $B$ represent {\em women}. Each man $a \in A$
possesses a (strict, and complete) preference order over women in $B$,
and similarly, all women in $B$ have a preference order over men in
$A$. A matching $M$ in~$G$ is called {\em stable} if there are no {\em
  blocking pairs} $(a,b)$; i.e. there do not exist $(a,b) \not \in M$
where both $a$ and $b$ prefer each other over their current partners
in~$M$ (if there are any). In their celebrated
work~\cite{GaleShapley62}, Gale and Shapley proposed an efficient
algorithm for finding a stable matching, providing a constructive proof
that stable matchings {\em always} exist.

Stable matchings have wide-spread applications (e.g., see
Manlove~\cite{Ma13}), and many of these are large-scale. Therefore, as
McDermid~\cite{McDermid} points out, assuming that
preferences are complete and strict is not realistic.  Thus, in this
paper, we will focus on stable matchings in the setting where
preference lists are allowed to be incomplete and contain ties. Here, a
woman is allowed to be indifferent between various men, and similarly,
a man may be indifferent between several women.  In this setting we
consider the {\em maximum-cardinality stable matching} problem where
the goal is to find a stable matching of maximum cardinality.

It is well-known that, in the settings where $G$ is either complete or
preferences do not contain ties, all stable matchings have the same
cardinality~\cite{GaleSotomayor85}. Moreover, a straightforward extension of the algorithm in
\cite{GaleShapley62} solves our problem in these cases. When ties and
incomplete preferences are permitted simultaneously, on the other hand, 
the problem of finding a maximum-cardinality stable
matching is well-known to be \np-hard~\cite{Manlove}. Furthermore, Yanagisawa~\cite{Ya07} showed that it is \np-hard to find a
$(33/29 - \varepsilon)$-approximate, maximum-cardinality stable
matching. The same author also showed
that assuming the {\em unique games} conjecture (UGC) it is hard to achieve
performance guarantee of $4/3 - \varepsilon$.

On the positive side, maximum-cardinality stable matchings with ties and incomplete preferences
have attracted significant attention
\cite{Bauckholt,HK15,Iwamaetal07,Iwamaetal14,Kiraly11,algorithms,Lam,LamPlaxton,mcdermid_first_approx,paluch_original}.
The best-known approximation algorithms for the problem achieve an
approximation ratio of
$3/2$~\cite{Kiraly11,mcdermid_first_approx,paluch_original}.

How does the hardness of maximum-cardinality stable matching depend on
the maximum allowed size of ties in the given instance? Huang and
Kavitha~\cite{HK15} recently considered the case where the {\em size
  of any tie} is bounded by $L=2$. The authors proposed an algorithm
and showed that its performance guarantee is at most $10/7$.  Chiang
and Pashkovich~\cite{CP18} later provided an improved analysis for the
same algorithm, showing that its real performance ratio is at most
$4/3$, and this result is tight under the UGC~\cite{Ya07}.  Lam
and Plaxton~\cite{LamPlaxton2019} very recently designed a
$1+(1-1/L)^L$-approximation algorithm for the so-called {\em
  one-sided} special case of our problem, where only preferences of
men are allowed to have ties.

\subsection{Our Contribution}

Our main result is captured in the following theorem. Note that the integrality gap of the natural LP relaxation for the problem is at least $(3L-2)/(2L-1)$~\cite{Iwama2014}. Hence, the performance ratio of our algorithm matches the known lower bound on the integrality gap.

\begin{restatable}{theorem}{mainthm}
  \label{thm:main result}
  Given an instance of the maximum-cardinality stable matching problem
  with incomplete preferences, and ties
  of size at most $L$; the polynomial-time algorithm described in Section \ref{algorithm} finds a
  stable matching $M$ with 
  \[
    |M|\ge \frac{2L-1}{3L-2}  \,|\OPT|,
  \]
  where $\OPT$ is an optimal stable matching.
\end{restatable}

Our algorithm is an extension of that by Huang and
Kavitha~\cite{HK15} for ties of size two: every man has $L$ proposals where each proposal goes to the acceptable women. Women can {\em accept} or
{\em reject} these proposals under the condition that no woman holds
more than $L$ proposals at any point during the algorithm.  Similar to
the algorithm in \cite{HK15}, we use the concept of {\em promotion} introduced by Kir{\'a}ly~\cite{Kiraly11} to
grant men repeat chances in proposing to women.  In comparison to
\cite{HK15}, the larger number of proposals in our algorithm leads to
subtle changes to the forward and rejection mechanisms of women, and to further
modifications to the way we obtain the output matching.

Our analysis is inspired by the analyses of both, Chiang and
Pashkovich~\cite{CP18}, and Huang and Kavitha~\cite{HK15}, 
but requires several new ideas to extend it to
the setting with larger ties. In both~\cite{HK15} and~\cite{CP18}, the analyses are based on \emph{charging schemes}: some objects are first assigned some values, called charges, and then charges are redistributed to nodes by a cost function. After a charging scheme is determined, relations between the generated total charges, and the sizes of output and optimal matchings are established, respectively, that lead to an approximation ratio.
The analysis in~\cite{HK15}
employs a complex charging scheme that acts {\em globally},
possibly distributing charges over the entire
graph. In contrast, the charging scheme in~\cite{CP18} is {\em
  local} in nature, and exploits only the local structure of the
output and optimal matchings, respectively.

We do not know of a direct way to extend the local {\em cost}-based
analysis of \cite{CP18} to obtain an approximation algorithm whose
performance beats the best known $3/2$-approximation for the general
case. Indeed we believe that
any such improvement must involve a non-trivial change in the charging
scheme employed. As a result, we propose a new analysis that combines
local and global aspects from \cite{CP18,HK15}. The central technical
novelty in the analysis is captured by Lemma~\ref{lem:cost_M_augmenting path} that provides an improved lower bound on the {\em cost} of components whereas Corollary~\ref{cor:cost_not_M_augmenting path} bounds the cost from below by
a simple multiple of the number of edges that are contained both in an optimal matching and in
the components.
 As we will see later, our new charging
scheme allows for a more fine-grained accounting of augmenting paths
for the output matching of our algorithm. 
\section{Algorithm for Two-Sided Ties of Size up to $L$}\label{algorithm}

We introduce some notational conventions. Let $a', a'' \in A$ be on the preference
list of $b \in B$. We write $a' \simeq_b a''$ if  $b$ is
\emph{indifferent} between $a'$ and $a''$, and we write 
$a'>_b a''$, or  $a' \geq_b a''$ if $b$ {\em strongly}, or {\em
  weakly} prefers $a'$ over $a''$, respectively. The preferences of
men over women are defined analogously. For $c \in A \cup B$, we let $N(c)$ denote the set of nodes adjacent to $c$ in~$G$.

\subsection{How men propose}

Each man $a\in A$ has $L$ proposals $p_a^1, p_a^2, \ldots, p_a^L$.
A man starts out as {\em basic}, and later becomes {\em 1-promoted}
before he is eventually elevated to {\em 2-promoted} status. Each man $a \in A$
has a \emph{rejection history} $R(a)$ which records the women who
rejected a proposal from $a$ during his current promotion status.
Initially, we let $R(a)=\varnothing$, for all $a \in A$. 

Each proposal $p_a^i$ for $a\in A$ and $i = 1,2,\ldots ,L$ goes to a
woman in $N(a)\setminus R(a)$ most preferred by $a$, and ties are
broken arbitrarily.  If a proposal $p_a^i$ for $a\in A$ and
$i = 1,2,\ldots,L$ is rejected by a woman $b\in B$, $b$ is added to
the rejection history of $a$, and subsequently, $p_a^i$ is sent to a
most preferred remaining woman in $N(a)\setminus R(a)$.

Suppose now that $R(a)$ becomes equal to $N(a)$ for some man
$a \in A$. If $a$ is either basic or 1-promoted then $a$'s rejection
history is cleared, and $a$ is promoted.
Otherwise, if $a$ is already 2-promoted, $a$ stops making proposals. 

\subsection{How women decide}

Each woman $b\in B$ can hold up to $L$ proposals, and among these
more than one can come from the same man. 
Whenever she holds less
than $L$ proposals, newly received proposals are automatically
accepted.  Otherwise, $b$ first tries to \emph{bounce} one of her
proposals, and if that fails, she will try to \emph{forward} one of her
proposals. If $b$ can neither bounce nor forward a
proposal, then $b$ rejects a proposal.

We continue describing the details. 
In the following, we let $P(b)$ and $A(b)$ denote the set of proposals
held by $b \in B$ at the current point, and the set of men
corresponding to these, respectively. Suppose that $|P(b)|=L$, and
that $b$ receives a new proposal $p_a^i$ for some $a\in A$ and
$i=1,\ldots, L$.

{\bf Bounce step.}  If there is a man $\alpha \in A(b) \cup \{a\}$ and
a woman $\beta\in B\setminus\{b\}$ such that
$\beta \simeq_{\alpha} b$, and $\beta$ currently holds less than $L$
proposals, then we move one of $\alpha$'s proposals from $b$ to $\beta$,
and we call the bounce step \emph{successful}.

{\bf Forward step.} If there is a man $\alpha \in A(b) \cup \{a\}$ and
a woman $\beta\in B \setminus \{b\}$ such that
$\beta \simeq_{\alpha} b$, at least two proposals from
$\alpha$ are present in $P(b)$, no proposal from $\alpha$ is present
in $P(\beta)$ and $\beta$ is not in $R(\alpha)$, then $b$ {\em forwards} a
proposal $p_{\alpha}^j \in P(b)\cup \{p^i_a\}$ for some $j=1,\ldots, L$ to $\beta$ and the
forward step is called \emph{successful}. As a consequence of a
successful forward step, $\alpha$ makes the proposal $p_{\alpha}^j$ to
$\beta$.

We point out that bounce and forward steps do not lead to an update to
the rejection history of an involved man. To describe the rejection
step, we introduce the following notions. For a woman $b\in B$,
a proposal $p_{a'}^{i'}$ is called \emph{more desirable} than
$p_{a''}^{i''}$ for $a',a''\in A$ and $i', i''=1,\ldots, L$ if $b$
strongly prefers $a'$ to $a''$, or if $b$ is indifferent between $a'$
and $a''$ and $a'$ has higher promotion status than $a''$.  A
proposal $p^{i'}_{a'} \in P(b)$ is \emph{least desirable} in $P(b)$
if $p^{i'}_{a'}$ is not more desirable than any proposal in $P(b)$.
Whenever $b \in B$ receives a proposal $p_a^i$, $|P(b)|=L$, and
neither {\em bounce} nor {\em forward} steps are successful, we execute a
rejection step.

{\bf Rejection step.}  If there is unique \emph{least desirable
  proposal} in $P(b) \cup \{p_a^i\}$, then $b$ rejects that proposal.
Otherwise, if there are more than one least desirable proposal in
$P(b)$, $b$ rejects a proposal from a man with the largest number of
least desirable proposals in $P(b) \cup \{p_a^i\}$. If
there are several such men, then we break ties arbitrarily.
Subsequently, $b$ is added to the rejection history of the man whose proposal
is rejected. 

\subsection{The algorithm}

An approximate maximum-cardinality stable matching for a given
instance $G = (A \cup B, E)$ is computed in two stages.

{\bf Stage 1.} Please see Algorithm \ref{stage1:pseudo}  for
the pseudo code for Stage 1. 

Men propose in an arbitrary order and women bounce,
forward or reject proposals as described above. The first stage
finishes, when for each man $a \in A$, one of the following two conditions
is satisfied: all proposals of $a$ are accepted; $R(a)$ becomes
equal to $N(a)$ for the third time.

We represent the outcome of the first stage as a bipartite graph 
$G' = (A \cup B, E')$ with the node set $A \cup B$ and the edge set
$E'$, where each edge $(a, b) \in E'$ denotes a proposal from $a$ held
by $b$ at the end of the first stage. Note that $G'$ may be a
multigraph in which an edge of the form $(a, b)$ appears with
multiplicity equal to the number of proposals that $b$ holds from $a$.
Clearly, each node $u$ in~$G'$ has degree at most $L$, denoted by
$\deg_{G'}(u) \le L$, since every man has at most $L$ proposals that
may be accepted and every woman can hold at most $L$ proposals at any
point in the first stage. 

\begin {algorithm}[!ht]
\caption{Pseudo code for Stage 1 of the algorithm\label{stage1:pseudo}}
\begin{algorithmic}[1]
\State let $G = (A \cup B, E)$ be an instance graph, and $N(c)$ denote the set of nodes adjacent to $c \in A\cup B$ in~$G$
\State let $G' = (A \cup B, E')$ be a multigraph with $E'$ initialized to the empty multiset of edges
\State let $\deg_{G'}(u)$ denote the degree of node $u$ in~$G'$, and $A(b)$ denote the set of nodes adjacent to $b \in B$ in~$G'$
\ForAll {$a \in A$}
\State $R(a) \coloneqq \varnothing$ \Comment $R(a)$ is the rejection history of man $a$
\State $stat_{a} \coloneqq 0$ \Comment $stat_a$ is the promotion status of man $a$
\EndFor
\While{$\exists a \in A$ s.t. $\deg_{G'}(a) < L$ and $R(a) \neq N(a)$}
\State let $b \in N(a) \setminus R(a)$ be a woman s.t. $b \ge_{a} b'$ for all $b' \in N(a) \setminus R(a)$
\State \Call {propose}{$a, b$}
\EndWhile
\State \textbf{return} $E'$
\Statex \{The following subroutine describes how $b$ accepts the proposal from $a$, or bounces, forwards, or rejects a proposal\}
\Procedure {propose}{$a, b$}
\If{$\deg_{G'}(b) < L$}
\State $E' \coloneqq E' \cup \{(a, b)\}$ 
\ElsIf {$\exists \alpha \in A(b) \cup \{a\}$ and $\exists \beta \in N(\alpha)$ s.t. $\beta \simeq_{\alpha} b$ and $\deg_{G'}\beta < L$}
\State $E' \coloneqq E' \cup \{(a, b), (\alpha, \beta)\} \setminus \{(\alpha, b)\}$\Comment bounce 
\ElsIf {$\exists\alpha \in A(b) \cup \{a\}$ and $\exists \beta \in N(\alpha) \setminus R(\alpha)$ s.t. $\beta \simeq_{\alpha} b$, 
\Indstatex $|(E' \cup \{(a, b)\}) \cap \{(\alpha, b)\}| \ge 2$ and $\alpha \notin A(\beta)$}
\State $E' \coloneqq E' \cup \{(a, b)\} \setminus \{(\alpha, b)\}$
	\State \Call {propose}{$\alpha, \beta$}\Comment forward
\Else		
\State let $\mathcal{A}$ denote $\{\alpha \in A(b) \cup \{a\} :$ for all $a' \in A(b) \cup \{a\}$, $\alpha \le_{b} a'$ and 
\Indstatex if $\alpha \simeq_{b} a'$, then $stat_{\alpha} \le stat_{a'}\}$ 
\State let $\alpha_0$ be a man in $\argmax_{\alpha \in \mathcal{A}} |(E' \cup \{(a, b)\}) \cap \{(\alpha, b)\}|$
\State $E' \coloneqq E' \cup \{(a, b)\} \setminus \{(\alpha_0, b)\}$\Comment reject
\State $R(\alpha_0) \coloneqq R(\alpha_0) \cup \{b\}$
	\If{$R(\alpha_0) = N(\alpha_0)$}
		\If{$stat_{\alpha_0} < 2$}
		\State $stat_{\alpha_0} \coloneqq stat_{\alpha_0} + 1$
		\State $R(\alpha_0) \coloneqq \varnothing$
		\EndIf
	\EndIf
\EndIf
\EndProcedure
\end {algorithmic}
\end {algorithm}

{\bf Stage 2.} We compute a maximum-cardinality matching $M$ in~$G'$
such that all nodes of degree $L$ in~$G'$ are matched.  The existence
of such matching is guaranteed by Lemma~\ref{lem:M_existence}. The
result of the second stage is such a matching $M$, that is the
output of the algorithm.

\begin{lemma}\label{lem:M_existence}
  There exists a matching in the graph $G'$ such that all nodes of
  degree $L$ in~$G'$ are matched. Moreover, there is such a matching
  $M$, where all nodes of degree $L$ in~$G'$ are matched and we have
	\[
	|M|\geq |E'|/L\,.
	\]
\end{lemma}
\begin{proof}
Consider the graph $G'=(A\cup B, E')$ and the following linear program
\begin{align*}
  \max ~~ &  \sum_{e\in E'} x_e   \\
  \st ~~ &
           \sum_{e\in \delta(u)}x_e \leq  1 ~~ (u\in A\cup B) \\
          & \sum_{e\in \delta(u)}x_e=1  ~~ (u\in A\cup B, \deg_{G'}(u)=L) \\
          & x\geq 0.
\end{align*}
It is well-known that the feasible region of the above LP is an
integral polyhedron. Moreover, the above LP is feasible as is easily
seen by considering the point that assigns $1/L$ to each edge in
$E'$. Hence there exists an integral point optimal for this linear
program. Notice, that every integral point feasible for this linear
program is a characteristic vector of a matching in~$G'$, which
matches all nodes of degree $L$ in~$G'$.  To finish the proof, notice
that the value of the objective function calculated at $x^\star$
equals $|E'|/L$. Thus the value of this linear program is at least
$|E'|/L$, finishing the proof.
\qed
\end{proof}

\subsection{Stability of output matching}

Let the above algorithm terminate with a matching $M$. We first argue that it is stable.

\begin{lemma}\label{lem:M_stability}
	The output matching $M$ is stable in $G=(A\cup B, E)$.
\end{lemma}
\begin{proof}
  Suppose for contradiction that $M$ is not stable,
  i.e. suppose that there exists an edge $(a,b) \in E$ that blocks
  $M$.  If $b$ rejected a proposal from $a$ during the algorithm, then
  $b$ holds $L$ proposals when the algorithm terminates and all these proposals are from
  men, who are weakly preferred by $b$ over $a$. Thus the degree of 
  $b$ in~$G'$ is $L$ implying that $b$ is matched in~$M$ with a man, who is
  not less preferred by $b$ than $a$. We get a contradiction to the
  statement that $(a,b)$ blocks $M$.
	
  Conversely, if $b$ did not reject any proposal from $a$ during the algorithm,
  then the algorithm terminates with all $L$ proposals of $a$ being accepted, particularly, by women, who are weakly preferred by $a$
  over $b$. Therefore the degree of $a$ in~$G'$ is $L$ implying that $a$ is
  matched in~$M$ with a woman, who is not less preferred by $a$ than
  $b$. Again, we get a contradiction to the statement that $(a,b)$ is a blocking pair for $M$.  
  \qed
\end{proof}

\subsection{Running time}

We now show that 
each stage of the algorithm has polynomial execution time. For the first stage, we illustrate that only a polynomial number of proposals are
bounced, forwarded, or rejected during this stage. For the second stage, the
proof of Lemma~\ref{lem:M_stability} implies that it is sufficient to find an optimal extreme solution for a linear program of
polynomial size.

First, we show that proposals are bounced only polynomially many times. For every $b \in B$,
at most $L$ proposals may be bounced to $b$. Indeed, with each
proposal bounced to $b$, the number of proposals held by $b$ increases;
also, the number of
proposals held by $b$ never decreases or exceeds $L$ during the algorithm. Hence at most $L |B|$
proposals are bounced during the first stage. 

Second, we illustrate that proposals are forwarded only polynomially
many times. For each $a\in A$, promotion status of $a$, and $b\in B$ such that $(a, b) \in E$, at most one proposal of $a$ may be forwarded to $b$. To see this, let $b'$ be a woman forwarding a proposal of $a$ to $b$. Notice that $b$ cannot bounce the proposal after $b$ receives it because, otherwise, $b'$ could bounce it by the transitivity of indifference. Observe also that $b$ may forward a proposal from $a$ only if she holds another proposal from him. Then it follows from the forward step that no woman can forward a proposal of $a$ to $b$ as long as $b$ holds a proposal from him. If $b$ rejects the proposal, then she is added to the rejection history of $a$, and so $b$ does
not receive any proposal from $a$ unless the promotion status of $a$
changes. Hence at most $3 |A| |B|$ proposals are forwarded during
the first stage.

Finally, for each $a\in A$, promotion status of $a$, and $b\in B$ such that $(a, b) \in E$, $b$ may reject at most $L$
proposals from $a$. Indeed, $b$ holds at most $L$ proposals at any point in time, and since $b$ is added to the rejection history of $a$ after she rejected him, $b$ does
not receive any proposal from $a$ unless the promotion status of $a$
changes. Hence at most $3 L |A| |B|$ proposals are rejected during
the first stage.

\section{Tight Analysis}

Recall that $\OPT$ is a maximum-cardinality stable matching in~$G$,
and let $M$ be the output matching defined above. If $a \in A$ is
matched with $b \in B$ in~$\OPT$, we write
$\OPT(a) \coloneqq b$ and $\OPT(b) \coloneqq a$. Similarly, we use the
notations $M(a) \coloneqq b$ and $M(b) \coloneqq a$ when $a \in A$ is
matched with $b \in B$ in~$M$. Note that our analysis is based on graph $G'$ and therefore all graph-related objects will assume $G'$.

\begin{definition}
A man $a \in A$ is called \emph{successful} if the algorithm terminates with all of his $L$
proposals being accepted. Likewise, a woman $b$ is called
\emph{successful} if she holds $L$ proposals when the algorithm stops. In other words, a person $c \in A \cup B$ is \emph{successful} if the degree of $c$ in~$G'$ is $L$, and \emph{unsuccessful} otherwise.
\end{definition}

\begin{definition} A woman is called \emph{popular} if she rejected a proposal during the algorithm, and \emph{unpopular} otherwise.
\end{definition}

Remarks~\ref{rem:bounce_step} and~\ref{rem:rejection_step} below directly follow from the algorithm and are consequences of the bouncing step, and the rejection step, respectively.

\begin{remark}\label{rem:bounce_step}
  Let $a \in A$ and $b,b'\in B$ be such that $b$ holds a
  proposal from $a$ when the algorithm finishes, $b'$ is unsuccessful, and $b' \simeq_a b$. Then
  $b$ is unpopular.
\end{remark}
\begin{proof}
  Suppose for contradiction that $b$ is popular. Then at some point
  she could not bounce or forward any one of her proposals and so she
  was to reject a proposal. This implies that after $b$ became
  popular, whenever she received a new proposal that could be bounced,
  that proposal would immediately be bounced. But then, when the
  algorithm terminates, $b$ holds a proposal from $a$, that could
  successfully be bounced to~$b'$, a contradiction.  \qed
\end{proof}

\begin{remark}\label{rem:rejection_step}
	Let $a,a' \in A$ and $b \in B$ be such that $b$ holds at least two
  proposals from $a$ when the algorithm finishes, $b$ rejected a proposal from $a'$ at some point, $a$ is basic, and $a' \simeq_b a$. Then there is an edge $(a', b)$ in~$G'$.
\end{remark}
\begin{proof}
	Suppose for a contradiction that $(a', b) \notin G'$ holds. Let $t$ be the most recent point in time when $b$ rejects a proposal from $a'$. Then it follows from the algorithm that, at $t$, $a'' \ge_{b} a'$ holds for all $a'' \in A(b)$. The rejection step also implies that, at $t$, there is no $a'' \in A$ such that $a' \simeq_b a''$, $a''$ is basic, and $b$ holds more than one proposal from $a''$. Moreover, the algorithm implies that, after $t$, whenever she receives a new proposal from a man $a''$ such that $a'' <_b a'$, she will immediately reject it unless she successfully bounces or forwards it. Now, consider a point in time after $t$ when there is a man $a''$ such that $a' \simeq_b a''$, $b$ already holds a proposal from $a''$, and receives another proposal from $a''$. Then the rejection step implies that she will reject one of the proposals from $a''$ unless she successfully bounces or forwards it. But then, when the algorithm terminates, $b$ holds at least two proposals from $a$, a contradiction.
\end{proof}

\subsection{Analytical techniques}

In the following, we define \emph{inputs}, \emph{outputs}, and
\emph{costs} -- notions that are central in the analysis of our charging scheme.
Before we take a closer look at these notions and
define them formally, let us discuss phenomena captured by them.

We use two different objects, inputs and outputs, to
differentiate between two different viewpoints on proposals accepted
when the algorithm ends. In particular, inputs are associated with
the viewpoint of women on the proposals whereas outputs are associated with the viewpoint of men. The choice of
terms ``inputs" and ``outputs" is due to the analysis in~\cite{HK15} where the edges of $G'$ are directed from men to women, and
so each proposal becomes an ``input" for the woman, and analogously becomes an
``output" for the corresponding man.

Now we describe the ideas that motivated our definitions concerning
outputs and inputs. Let $M + \OPT$ denote the multiset that contains the edges in~$M$ and the edges in~$\OPT$. To establish the approximation guarantee of our
algorithm, we analyze each connected component in $M + \OPT$. In order
to show that $M$-augmenting paths in $M + \OPT$ do not lead to a large
approximation guarantee, we introduce the notions of \emph{bad} and
\emph{good inputs} as well as \emph{bad} and \emph{good outputs}. For
example, a certain number of bad inputs and bad outputs are generated
by the edges incident to the endpoints of an $M$-augmenting path in $M + \OPT$. Indeed, as we will see later, if $a_0-b_0-a_1- \dotsc -a_k-b_k$ is an $M$-augmenting path in
$M + \OPT$ of length $2k +1$, $k \ge 2$ where $a_0 \in A$, then $b_0$ has at least $L-2$
bad inputs and $a_k$ has at least $L-2$ bad outputs. Then to show the
approximation guarantee of $(3L-2)/(2L-1)$, we provide a way to obtain a lower
bound on the number of bad inputs and bad outputs of men and women in
each $M$-augmenting path; and later we provide an upper bound on the
total number of bad inputs and bad outputs of all men and women.

To implement the above ideas, we use a charging scheme. Our charging scheme
associates a cost with each man and each woman. These costs
keep track of bad inputs and bad outputs: bad inputs lead to an increase of the corresponding woman's
cost and bad outputs lead to an increase of the corresponding man's
cost. We show that the total cost of all men and women is bounded above by
$2L|M|$. On the other side, we provide a lower bound
on the total cost by giving a lower bound on the cost of each connected component in $M + \OPT$. These upper
and lower bounds lead to the desired approximation guarantee of
$(3L-2)/(2L-1)$.

\subsection{Inputs and outputs}
In our analysis inputs and outputs are fundamental edge-related objects for our charging scheme. Each edge in~$G'$ generates a certain number of charges. For example, as we will see in Section~\ref{cost}, if an edge $(a, b)$ in~$G'$ belongs either to $M$ or to $\OPT$, two charges are generated by $(a, b)$ so that one is carried to node $a$ and one is carried to node $b$ by cost function. To define similar charging mechanisms for the remaining types of edges in~$G'$, we first distinguish them as in the following definitions. 

\begin{definition}
Given an edge $(a, b)$ in~$G'$, we say that $(a, b)$ is an \emph{output from} $a \in A$ and an \emph{input to} $b \in B$ if $(a, b)$ is not in $M + \OPT$. 
\end{definition}

To illustrate how outputs and inputs are determined, for example,  let $(a,b)\in M$, $a\in A$, $b\in B$ and $n_{(a,b)}$ be  the number of edges of the form $(a,b)$ in the multigraph $G'$, then the edge $(a,b)$ gives rise to the following number $s_{(a,b)}$ of inputs (and to the same number of outputs)
\[
s_{(a,b)}:=\begin{cases}
			n_{(a,b)}-1 &\text{if }(a,b)\not\in \OPT\\
			0		&\text{if }n_{(a,b)}=1\\
			n_{(a,b)}-2 &\text{otherwise }\,.\\
		\end{cases}
\]

\begin{definition}\label{bad_good_input_output}
An input $(a, b)$ to $b \in B$ is called  a \emph{bad input} if one of the following is true:
\begin{itemize}
\item $b$ is popular and $a >_b \OPT(b)$.
\item $b$ is popular, $a \simeq_b \OPT(b)$, but $\OPT(b)$ is unsuccessful.
\item $b$ is popular, $a$ is 1-promoted, $\OPT(b)$ is successful and $M(b) \simeq_b \OPT(b) \simeq_b a$.
\end{itemize}

An input $(a, b)$ to $b \in B$ is a \emph{good input} if it is not a bad input. In other words, an input $(a, b)$ to $b \in B$ is a \emph{good input} if one of the following is true:
\begin{itemize}
\item $b$ is unpopular.
\item $b$ is popular and $\OPT(b) >_b a$.
\item $b$ is popular, $a \simeq_b \OPT(b)$, $\OPT(b)$ is successful and $a$ is not 1-promoted.
\item $b$ is popular, $a \simeq_b \OPT(b)$, $\OPT(b)$ is successful, but not $M(b) \simeq_b \OPT(b) \simeq_b a$.
\end{itemize}	

An output $(a, b)$ from a man $a$ is called a \emph{bad output} if one of the following is true:
\begin{itemize}
\item $b$ is unpopular.
\item $b$ is popular, $b >_a \OPT(a)$, $a$ is 1-promoted, but not $M(b) \simeq_b \OPT(b) \simeq_b a$.
\item $b$ is popular, $b >_a \OPT(a)$ and $a$ is basic.
\end{itemize}	

An output from a man $a$ is a \emph{good output} if that is not a bad output. In other words, an output $(a, b)$ from a man $a \in A$ is a \emph{good output} if one of the following is true:
\begin{itemize}
\item $b$ is popular and $\OPT(a) \ge_a b$.
\item $b$ is popular, $b >_a \OPT(a)$ and $a$ is 2-promoted.
\item $b$ is popular, $b >_a \OPT(a)$, $a$ is 1-promoted and $M(b) \simeq_b \OPT(b) \simeq_b a$.
\end{itemize}	
\end{definition}

\begin{lemma}\label{lem:no_bad_input_bad_output}
	There is no edge which is both a bad input and a bad output.
\end{lemma}

\begin{proof}
Assume that an edge $(a, b)$, $a \in A$, $b \in B$ is both a bad input to $b$ and a bad output from $a$. First, consider the first case from the definition of a bad output. It trivially contradicts all the cases from the definition of a bad input. Second, consider the first case from the definition of a bad input and either the second or the third case from the definition of a bad output. Then the case (\ref{case:no_bad_input_bad_output1}) below is implied. Third, consider the second case from the definition of a bad input and either the second or the third case from the definition of a bad output. Then the case (\ref{case:no_bad_input_bad_output2}) below is implied. Finally, consider the third case from the definition of a bad input. It trivially contradicts both the second and the third case from the definition of a bad output. Thus one of the following cases is true:
\begin{enumerate}
    \item \label{case:no_bad_input_bad_output1} $a >_b \OPT(b)$; $b >_a \OPT(a)$.
    \item \label{case:no_bad_input_bad_output2} $a \simeq_b \OPT(b)$, and $\OPT(b)$ is unsuccessful; $a$ is not 2-promoted.
\end{enumerate}

In case~\eqref{case:no_bad_input_bad_output1}, the edge $(a, b)$ is a blocking pair for $\OPT$, contradicting the stability of $\OPT$. 

In case~\eqref{case:no_bad_input_bad_output2}, since $\OPT(b)$ is unsuccessful, $\OPT(b)$ was rejected by $b$ as a 2-promoted man. On the other hand, $a \simeq_b \OPT(b)$, $a$ is not 2-promoted, and $b$ holds a proposal from $a$ when the algorithm terminates, contradicting the rejection step.
\qed
\end{proof}

\begin{corollary}\label{cor:number_good_inputs_bad_outputs}
	The number of good inputs is at least the number of bad outputs.
\end{corollary}

\begin{proof}
	Assume for a contradiction that the number of good inputs is smaller than the number of bad outputs. Then there is an edge in~$G'$ which is a bad output but not a good input. In other words, there is an edge in~$G'$ which is both a bad output and a bad input, contradicting Lemma~\ref{lem:no_bad_input_bad_output}.
\qed
\end{proof}

\subsection{Cost}\label{cost}

In our charging scheme, cost is a function that assigns  charges, that originate from the edges, to the nodes. More specifically, the cost of a man $a$ is obtained by counting the edges in~$G'$ incident to $a$, where bad outputs contribute~$2$ and all other edges contribute~$1$. Similarly, the cost of a woman $b$ is obtained by counting the edges in~$G'$ incident to $b$, to which good inputs contribute~$0$ and all other edges contribute~$1$.

In the following, let $\deg(u)$ be the degree of the node $u$ in~$G'$. For  $a\in A$, we define his \textit{cost} as follows:
$$\cost(a) \coloneqq \deg(a) + k\,, \ \text{ where $k$ is the number of bad outputs from $a$};$$
\noindent for $b\in B$, we define her cost as follows:
$$\cost(b) \coloneqq \deg(b) - k\,, \ \text{ where $k$ is the number of good inputs to $b$},$$

\noindent For a node set $S \subseteq A \cup B$, $\cost(S)$ is defined as the sum of costs of all the nodes in $S$.
 
The above definitions lead to next three remarks.

\begin{remark}\label{rem:bad_inputs_cost1}
	Let $b \in B$ be matched in~$M$ and have at least $k$ bad inputs. Then $\cost(b) \ge k + 1$.
\end{remark}
\begin{proof}
	Let $k'$ be the number of good inputs to $b$. Since $b$ is matched in~$M$, the edge $(M(b), b)$ is contained in~$G'$ and therefore it is not an input to $b$. Thus $\deg(b) \ge k + k' + 1$. Hence, by definition of cost, $\cost(b) = \deg(b) - k' \ge k + 1$ holds. 
\qed
\end{proof}

\begin{remark}\label{rem:bad_inputs_cost2}
	Let $b \in B$ be matched in~$\OPT$, have at least $k$ bad inputs, and $(\OPT(b),b)\in E'$ where $E'$ is the edge set of $G'$. Then $\cost(b) \ge k + 1$.
\end{remark}
\begin{proof}
	Let $k'$ be the number of good inputs to $b$. Since the edge $(\OPT(b), b)$ is in~$G'$, it is not an input to $b$. Thus $\deg(b) \ge k + k' + 1$. So, by definition of cost, $\cost(b) = \deg(b) - k' \ge k + 1$ holds. 
\qed
\end{proof}

\begin{remark}\label{rem:M_and_OPT_cost}
	Let $b \in B$ be matched in both $\OPT$ and $M$, $\OPT(b) \neq M(b)$, and $(\OPT(b),b)\in E'$ where $E'$ is the edge set of $G'$. Then $\cost(b) \ge 2$.
\end{remark}
\begin{proof}
	Let $k$ and $k'$ be the numbers of bad inputs and good inputs to $b$, respectively. Since the edges $(\OPT(b), b)$ and $(M(b), b)$ are contained in~$G'$, they are not inputs to $b$. Thus $\deg(b) \ge k + k' + 2$. So, by definition of cost, $\cost(b) = \deg(b) - k' \ge k + 2 \ge 2$ holds. 
\qed
\end{proof}

\subsection{The approximation ratio}

Let $\mathcal{C}(M + \OPT)$ denote the set of connected components in
a graph induced by the edge set $M +
\OPT$. Lemma~\ref{lem:cost_M_augmenting path} below bounds the cost of
$M + \OPT$. Because of space constraints, its proof is deferred to
Appendix \ref{main-lem}.

\begin{restatable}{lemma}{mainlem}
\label{lem:cost_M_augmenting path}
	$\sum_{C \in \mathcal{C}(M + \OPT)} \cost(C) \ge (L+1)|\OPT| + (L-2)(|\OPT|-|M|)$.	
\end{restatable}

We are ready to prove our main theorem, and restate it here for
completeness.

\mainthm*

\begin{proof}
	By Lemma~\ref{lem:M_existence}, we have
	\[
	|M| \ge \frac{|E'|}L = \sum_{u \in A \cup B} \frac{\deg(u)}{2L}\,.
	\]
By definition of cost and by Corollary~\ref{cor:number_good_inputs_bad_outputs}, we obtain 
	\[
	\sum_{u \in A \cup B} \deg(u) \ge \cost(A \cup B)\,.
	\]
Combining the above inequalities, we get	
	\[
	2L|M| \ge \sum_{u \in A \cup B} \deg(u) \ge \cost(A \cup B) = \sum_{C \in \mathcal{C}(M + \OPT)}\cost(C)\,,
	\]
By Lemma~\ref{lem:cost_M_augmenting path},  we obtain 	
	\[
	2L|M| \ge \sum_{C \in \mathcal{C}(M + \OPT)}\cost(C) \ge (L+1)|\OPT| + (L-2)(|\OPT|-|M|)\,.
	\]
 By rearranging the terms, we obtain 
	\[
	2L |M|	+(L-2) |M|\ge (L+1)|\OPT| +(L-2)|\OPT|\,,
	\]
	and so we obtain the desired inequality
	\[
	(3L-2) |M|\ge (2L-1)|\OPT|\,.
	\]
\qed
\end{proof}

\subsection{Costs of connected components in $M + \OPT$}\label{components of M+OPT}

The purpose of this subsection is to prove
Lemma~\ref{lem:cost_M_augmenting path}. We call a connected component
of $M + \OPT$ \emph{trivial} if it is an isolated node. A component in
$M + \OPT$ is called \emph{alternating path} if the sequence of its
edges alternate being contained in~$M$ and in~$\OPT$. An alternating
path is called \emph{alternating cycle} if its endpoints are the
same. We call an alternating path $\OPT$-augmenting if the edges
incident to its endpoints are in~$M$. Likewise, we call an alternating
path $M$-augmenting if the edges incident to its endpoints are in
$\OPT$.  For ease of exposition, henceforth, we will refer by
alternating paths only to the components that are not alternating
cycles, $\OPT$-augmenting or $M$-augmenting paths.

We begin by studying costs of connected components in $M + \OPT$. For each connected component, we find an appropriate lower bound. The costs of components that are alternating paths, alternating cycles or $\OPT$-augmenting paths, can be bounded from below by $L+1$ multiplied by the number of edges that are both in~$\OPT$ and in the associated component. However, the costs of $M$-augmenting paths can be bounded from below in a stronger way. While the costs for trivial paths, alternating paths, alternating cycles or $\OPT$-augmenting paths can be obtained in a straightforward way, those for $M$-augmenting paths are central to our analysis and require a detailed study. After we establish the lower bounds on the costs of all connected components in $M + \OPT$, we start proving Lemma~\ref{lem:cost_M_augmenting path}.

The following lemma bounds costs of edges in~$\OPT$ from below. Recall that $\deg(u)$ is the degree of the node $u$ in~$G'$.

\begin{lemma}\label{lem:cost_OPT_edge}
	Let $a \in A$ and $b \in B$ be such that $(a, b) \in \OPT$. Then $\cost(\{a, b\}) \ge L$ holds. Furthermore, if $\deg(a) \ge 1$, then $\cost(\{a, b\}) \ge L + 1$; if $\deg(b) \le L-1$, then $\cost(\{a, b\}) \ge 2L - 1$.
\end{lemma}

\begin{proof}
	We consider $\deg(a)$ and $\deg(b)$ simultaneously. Since both are integers between $0$ and $L$, the following cover all possible cases for values of $\deg(a)$ and $\deg(b)$:	
	\begin{enumerate}
		\item\label{case:cost_OPT_edge1} $\deg(a) = 0$ and $\deg(b) = L$.
		\item\label{case:cost_OPT_edge2} $1 \le \deg(a) \le L-1$ and $\deg(b) = L$.
		\item\label{case:cost_OPT_edge3} $\deg(a) = L$ and $\deg(b) = L$.
		\item\label{case:cost_OPT_edge4} $\deg(a) = L$ and $\deg(b) \le L-1$.
		\item\label{case:cost_OPT_edge5} $\deg(a) \le L-1$ and $\deg(b) \le L-1$.
	\end{enumerate}
	
	In cases~\eqref{case:cost_OPT_edge1}  and~\eqref{case:cost_OPT_edge2}, $a$ is unsuccessful. Since $(a, b)$ is an edge in~$G$, $b$ rejected a proposal from $a$, and so $b$ is popular. Thus there are $L$ separate edges $(a^1, b), (a^2, b), \dotsc (a^L, b)$ in~$G'$ such that $a \le_b a^i$ for all $i =1,2,\ldots,L$. Moreover, none of the edges $(a^1, b)$, $(a^2, b)$, \ldots, $(a^L, b)$ is a good input because $b$ is popular, $\OPT(b) \le_b a^i$ for all $i =1,2,\ldots,L$, and $\OPT(b)$ is unsuccessful. Thus $\cost(b) = L$.

Hence, for case~\eqref{case:cost_OPT_edge1},
	\[
		\cost(\{a, b\}) \ge \cost(b) = L\,;
	\]
for case~\eqref{case:cost_OPT_edge2}
	\[
	\cost(\{a, b\}) = \underbrace{\cost(a)}_{\ge \deg(a) \ge 1} + \underbrace{\cost(b)}_{=L} \ge L + 1\,,
	\]
	as required.
	
	In case~\eqref{case:cost_OPT_edge3}, $b$ is matched in~$M$ since $\deg(b) = L$. Thus, by Remark~\ref{rem:bad_inputs_cost1}, $\cost(b) \ge 1$ holds. Hence 
	\[
	\cost(\{a, b\}) = \underbrace{\cost(a)}_{\ge \deg(a) = L} + \underbrace{\cost(b)}_{\ge 1} \ge L + 1\,,
	\]
	as desired.
	
	In case~\eqref{case:cost_OPT_edge4}, $b$ is unsuccessful, and so $b$ did not reject any proposal from $a$. Thus, $a$ is basic and for every edge $(a, b') \in G'$ with $b'\in B$, and so $b' \ge_a b$. Thus, by Remark~\ref{rem:bounce_step} and Definition \ref{bad_good_input_output}, each edge $(a, b') \in G'$ with $b'\in B$ is a bad output from~$a$. 
	
	 Since $\deg(a) = L$, $a$ is matched in~$M$. Thus if there is no edge $(a, b)$ in~$G'$, then $a$ has $L-1$ bad outputs implying the desired inequality
	\[
	\cost(\{a, b\}) \ge \cost(a) = \deg(a) + L-1 = 2L-1\,.
	\]
But if there is an edge $(a, b)$ in~$G'$, then $a$ has $L-2$ bad outputs. Also $\cost(b) \ge 1$ by Remark~\ref{rem:bad_inputs_cost2}. Thus	
	\[
	\cost(\{a, b\}) = \underbrace{\cost(a)}_{\ge 2L-2} + \underbrace{\cost(b)}_{\ge 1} \ge 2L - 2 + 1 = 2L-1 \,,
	\]
	as needed.
	
	In case~\eqref{case:cost_OPT_edge5}, both $a$ and $b$ are unsuccessful. Since $(a, b)$ is an edge in~$G$ and $a$ is unsuccessful, a proposal from $a$ was rejected by $b$ at some time during the algorithm. On the other hand, since $b$ is unsuccessful, she did not reject any proposal during the algorithm, a contradiction.
\qed
\end{proof}

 For completeness, we state the following remark that is trivially true.
\begin{remark}\label{rem:cost_trivial_component}
Trivial components have cost at least $(L+1)|OPT \cap C|$.
\end{remark}

\subsubsection*{Alternating paths, alternating cycles and $\OPT$-augmenting paths} Recall that despite the original definition of alternating paths, we merely mean by them the components that are not alternating cycles, $\OPT$-augmenting or $M$-augmenting paths.
The following corollary of Lemma~\ref{lem:cost_OPT_edge} provides lower bounds on the costs of alternating paths, alternating cycles and $\OPT$-augmenting paths.
\begin{corollary}\label{cor:cost_not_M_augmenting path}
	Let $C$ be a connected component of $M + \OPT$ such that it is an alternating path, alternating cycle, or $\OPT$-augmenting path. Then $\cost(C) \ge (L+1)|\OPT \cap C|$.
\end{corollary}
\begin{proof}
	First, we note that since the length of an alternating path is even, the endpoints of it are either both men or both women as in (\ref{case:cost_not_M_augmenting path3}) and (\ref{case:cost_not_M_augmenting path4}) below. In contrast, the length of an $\OPT$-augmenting path is odd, and so its endpoints are a man and a woman as in (\ref{case:cost_not_M_augmenting path1}) below. Last, alternating cycles have the general form as in (\ref{case:cost_not_M_augmenting path2}) below, but it can be represented by various ways simply by shifting the nodes to the right or to the left. Assuming $C$ is as stated above, one of the following is true:
	\begin{enumerate}
		\item\label{case:cost_not_M_augmenting path1} $C$ is an OPT-augmenting path of the form $a_0-b_1-\dotsc - a_k-b_{k+1}.$
		\item\label{case:cost_not_M_augmenting path2} $C$ is an alternating cycle of the form $a_1-b_1-\dotsc - a_k-b_k - a_1,$ where $(a_1, b_1) \in \OPT$.
		\item\label{case:cost_not_M_augmenting path3} $C$ is an alternating path of the form $b_1-a_1-\dotsc - b_k-a_k - b_{k+1},$ where $a_1 \in A$ and $(a_1, b_1) \in \OPT$.
		\item\label{case:cost_not_M_augmenting path4} $C$ is an alternating path of the form $a_1-b_1-\dotsc - a_k-b_k - a_{k+1},$ where $a_1 \in A$ and $(a_1, b_1) \in \OPT$.
	\end{enumerate}
	
	For cases~\eqref{case:cost_not_M_augmenting path1}, \eqref{case:cost_not_M_augmenting path2} and~\eqref{case:cost_not_M_augmenting path3}, Lemma~\ref{lem:cost_OPT_edge} implies that $\cost(\{a_i, b_i\}) \ge L+1$ for every $i = 1,\ldots,k$. Thus
	\[
	\cost(C) \ge \sum_{i=1}^k \underbrace{\cost(\{a_i, b_i\})}_{\ge L+1} \ge (L+1)k = (L+1)|\OPT \cap C|\,,
	\]
	as required.
	
	For case~\eqref{case:cost_not_M_augmenting path4}, Lemma~\ref{lem:cost_OPT_edge} implies that $\cost(\{a_i, b_i\}) \ge L+1$ for every $i = 2,\ldots,k$ and $\cost(\{a_1, b_1\}) \ge L$. Since $a_{k+1}$ is matched in~$M$, $\cost(a_{k+1}) \ge 1$ holds. Thus
	\[
	\cost(C) \ge \underbrace{\cost(\{a_1, b_1\})}_{\ge L} + \underbrace{\cost(a_{k+1})}_{\ge 1} +  \sum_{i=2}^k \underbrace{\cost(\{a_i, b_i\})}_{\ge L+1} \ge (L+1)k = (L+1)|\OPT \cap C|\,.
	\]
\qed	
\end{proof}
			
\subsubsection*{$M$-augmenting paths}
In this section, we provide a lower bound on the cost of components in $M+\OPT$, that are $M$-augmenting paths of length at least $5$. We call an edge in an $M$-augmenting path \emph{terminal} if it is incident to either endpoint of the path, and \emph{internal} otherwise. We start by showing that there are no $M$-augmenting paths in $M+\OPT$ of length $1$ or $3$.

 \begin{lemma}\label{lem:no_M_augmenting_path_length_1_or_3}
	There is no $M$-augmenting path in $M + \OPT$, that is of length $1$ or of length $3$.
\end{lemma}

\begin{proof}
	First, suppose that there is an $M$-augmenting path in $M+\OPT$, that is of length~$1$. That is to say, there exists an edge $(a,b)$ in~$\OPT$ such that neither $a$ nor $b$ is matched in~$M$. Since $(a,b)$ is in~$G$ and none of $a$ and $b$ is matched in~$M$, $(a,b)$ is a blocking pair for $M$, that contradicts Lemma~\ref{lem:M_stability}.
	
	Second, suppose that there is an $M$-augmenting path in $M+\OPT$, that is of length~$3$ and of form $a_0-b_0-a_1-b_1$ where $a_0 \in A$. Since $a_0$ and $b_1$ are unmatched in~$M$, $\deg(a_0) < L$ and $\deg(b_1) < L$ hold, and hence both $a_0$ and $b_1$ are unsuccessful. Since $a_0$ is unsuccessful, he is 2-promoted and was rejected by every woman in his preference list as a 2-promoted man. Since $b_0$ is such a woman, she is popular. Also, we notice that $(a_1,b_0)$ is in~$M$, and hence $b_0$ holds a proposal from $a_1$ when the algorithm terminates. Thus $a_1 \ge_{b_0} a_0$.
	
	Observe that $(a_1,b_1)$ is in~$\OPT$, and hence $(a_1,b_1)$ is in~$G$. Since $b_1$ is unsuccessful, $b_1$ did not reject any proposal during the algorithm. Since no proposal from $a_1$ was rejected by $b_1$, he is basic. Also, $b_0 \ge_{a_1} b_1$ holds since $b_0$ holds a proposal from $a_1$ when the algorithm finishes and no proposal from $a_1$ was rejected by $b_1$. Thus $a_1 \ge_{b_0} a_0$, $b_0 \ge_{a_1} b_1$, $a_0$ is 2-promoted, $b_0$ is popular, $a_1$ is basic and $b_1$ is unsuccessful. 
	
	First, $a_1 \simeq_{b_0} a_0$ cannot hold because $a_1$ is basic, and $b_0$ rejected $a_0$ as a 2-promoted man, whereas $b_0$ holds a proposal from $a_1$ when the algorithm ends. Second, $b_0 \simeq_{a_1} b_1$ cannot hold, otherwise we get a contradiction to Remark~\ref{rem:bounce_step} since $b_0$ is popular, $b_1$ is unsuccessful, and $b_0$ holds a proposal from $a_1$ when the algorithm terminates. Hence we conclude that $a_1 >_{b_0} a_0$ and $b_0 >_{a_1} b_1$ hold. Since $(a_0, b_0) \in \OPT$ and $(a_1, b_1) \in \OPT$, $(a_1,b_0)$ is a blocking pair for $\OPT$, contradicting the stability of $\OPT$.
\qed
\end{proof}

Now, we consider $M$-augmenting paths in $M + \OPT$, that are of lengths at least $5$. Since the length of an $M$-augmenting path is odd, its endpoints are a man and a woman. Note that, our next results assume the representation, where, without loss of generality, the leftmost node is a man. In the following definition, a woman in an $M$-augmenting path points right is the compact way to say that the woman weakly prefers the man on her right to the man on her left, where the weakly preferred man is promoted if she is indifferent between them. 
\begin{definition}\label{def1}
	Let $a_0 - b_0 - a_1 - \ldots - a_k - b_k$ be an $M$-augmenting path of length at least $5$, where $a_0 \in A$. For $i = 0, \ldots, k-1$, we say that $b_i$ \emph{points right} if one of the following is true:
	\begin{itemize}
		\item $a_{i+1} >_{b_i} a_i$.
		\item $a_{i+1} \simeq_{b_i} a_i$, and $a_{i+1}$ is not basic.
	\end{itemize}
\end{definition}

The desired lower bound on the cost of $M$-augmenting paths in $M + \OPT$ is demonstrated by partitioning an $M$-augmenting path into the pieces of the first terminal edge, internal edges, and the last terminal edge, and providing a lower bound on the cost of each piece. 

Remarks~\ref{rem:cost_M_augmenting_start} and~\ref{rem:cost_M_augmenting_end_small} below provide bounds on the costs of the terminal edges of an $M$-augmenting path in $M + \OPT$.

\begin{remark}\label{rem:cost_M_augmenting_start}
	Let $a_0-b_0-a_1- \dotsc -a_k-b_k$ be an $M$-augmenting path in $M + \OPT$ of length $2k +1$, $k \ge 2$, where $a_0 \in A$. Then $\cost(\{a_0, b_0\}) \ge L$. Moreover, $b_0$ rejected a proposal from $a_0$ at some point, and $b_0$ points right.
\end{remark}
\begin{proof}
	First, since $(a_0, b_0) \in \OPT$, Lemma~\ref{lem:cost_OPT_edge} implies $\cost(\{a_0, b_0\}) \ge L$.
	Second, observe that $a_0$ is not matched in~$M$, and hence $a_0$ is unsuccessful. Thus $b_0$ rejected $a_0$ as a 2-promoted man. On the other hand, since $b_0$ has a proposal from $a_1$ when the algorithm finishes, we deduce that $a_1 \ge_{b_0} a_0$ holds. Notice that if $a_1 \simeq_{b_0} a_0$ holds, then $a_1$ is not basic. Thus $b_0$ points right, that finishes the proof.
\qed
\end{proof}

\begin{remark}\label{rem:cost_M_augmenting_end_small}
	Let $a_0-b_0-a_1- \dotsc -a_k-b_k$ be an $M$-augmenting path in $M + \OPT$ of length $2k +1$, $k \ge 2$, where $a_0 \in A$. Then $\cost(\{a_k, b_k\}) \ge 2L-1$.
\end{remark}
\begin{proof}
	Observe that $b_k$ is not matched in~$M$, and hence $\deg(b_k) \le L-1$ holds. Since $(a_k, b_k) \in \OPT$ and $\deg(b_k) \le L-1$, Lemma~\ref{lem:cost_OPT_edge} implies the desired inequality that $\cost(\{a_k, b_k\}) \ge 2L-1$.
\qed
\end{proof}

 Lemma~\ref{lem:cost_M_augmenting path_edge} below is important for a better understanding of the internal edges in~$M$-augmenting paths and can be considered as rather a technical result followed by a corollary that is of an essential use. The proof of Lemma~\ref{lem:cost_M_augmenting path_edge} is presented after we establish the key result of this section in Lemma~\ref{lem:cost_M_augmenting} and prove Lemma~\ref{lem:cost_M_augmenting path}.

\begin{lemma}\label{lem:cost_M_augmenting path_edge}	
	Let $a_0-b_0-a_1- \dotsc -a_k-b_k$ be an $M$-augmenting path in $M + \OPT$ of length $2k +1$, $k \ge 2$, where $a_0 \in A$. Then for every $i = 1, \dots ,k-1$, at least one of the following is true:
	\begin{enumerate}
		\item\label{case:cost_M_augmenting path_edge1} $\cost(\{a_i, b_i\}) \ge L+2$.
		\item\label{case:cost_M_augmenting path_edge2} $b_i$ rejected a proposal from $a_i$ at some point, and $b_i$ points right.
		\item\label{case:cost_M_augmenting path_edge3} $a_i$ is basic and $b_{i-1} >_{a_i} b_i$.
	\end{enumerate}
\end{lemma}

For an $M$-augmenting path in $M + \OPT$ of length $2k+1$, $k \ge 2$, Lemma~\ref{lem:cost_OPT_edge} implies that each internal edge that is both in the same path and in~$\OPT$ has cost at least $L+1$. The following corollary of Lemma~\ref{lem:cost_M_augmenting path_edge} establishes an essential fact when the cost of such an internal edge is exactly $L+1$.

\begin{corollary}\label{cor:cost_M_augmenting_relationship}
	Let $a_0-b_0-a_1- \ldots -a_k-b_k$ be an $M$-augmenting path in $M + \OPT$ of length $2k +1$, $k \ge 2$, where $a_0 \in A$. For every $i = 1, \dots ,k-1$ such that $\cost(\{a_i, b_i\}) = L+1$, if $b_{i-1}$ rejected a proposal from $a_{i-1}$ at some point and $b_{i-1}$ points right, then  $b_{i}$ rejected a proposal from $a_i$ at some point and $b_i$ points right.
\end{corollary}

\begin{proof}
	By Lemma~\ref{lem:cost_M_augmenting path_edge}, for every $i = 1, \dots ,k-1$, at least one of the following is true:
	\begin{enumerate}
 		\item \label{case:cost_M_augmenting_relationship1} $\cost(\{a_i, b_i\}) \ge L+2$.
 		\item \label{case:cost_M_augmenting_relationship2}  $b_i$ rejected a proposal from $a_i$ at some point, and $b_i$ points right.
 		\item \label{case:cost_M_augmenting_relationship3} $a_i$ is basic and $b_{i-1} >_{a_i} b_i$.
 	\end{enumerate}
 	
 In case~\eqref{case:cost_M_augmenting_relationship1}, that is an immediate contradiction to $\cost(\{a_i, b_i\}) = L+1$. 
 
 In case~\eqref{case:cost_M_augmenting_relationship3}, $a_i$ is basic. Thus if $b_{i-1}$ points right as stated, then $a_{i-1} <_{b_{i-1}} a_i$. Hence $a_{i-1} <_{b_{i-1}} a_i$ and $b_{i-1} >_{a_i} b_i$, showing that $(a_i, b_{i-1})$ is a blocking pair for $\OPT$, a contradiction to the stability of $\OPT$.
 
 In case~\eqref{case:cost_M_augmenting_relationship2}, we  obtain the desired statement.
\qed	
\end{proof}

Lemma~\ref{lem:cost_M_augmenting_end_nodes_relationship} below provides a bound on the cost of the rightmost internal edge of an $M$-augmenting path in $M + \OPT$ given the fact that is established by Corollary~\ref{cor:cost_M_augmenting_relationship} occurs. The proof of Lemma~\ref{lem:cost_M_augmenting_end_nodes_relationship} is presented after the proof of Lemma~\ref{lem:cost_M_augmenting path_edge}.

\begin{lemma}\label{lem:cost_M_augmenting_end_nodes_relationship}
	Let $a_0-b_0-a_1- \dotsc -a_k-b_k$ be an $M$-augmenting path in $M + \OPT$ of length $2k +1$, $k \ge 2$, where $a_0 \in A$. If $b_{k-1}$ rejected a proposal from $a_{k-1}$, and $b_{k-1}$ points right, then $\cost(\{a_{k-1}, b_{k-1}\}) \ge L+2$.
\end{lemma}	

Now, we have all the tools to bound the cost of $M$-augmenting paths of length at least $5$.
\begin{lemma}\label{lem:cost_M_augmenting}
	Let $C$ be a connected component of $M + \OPT$, that is an $M$-augmenting path of length at least~$5$. Then $\cost(C) \ge (L+1)|\OPT\cap C| + (L-2)$.
\end{lemma}
\begin{proof}
	Let $C$ be an $M$-augmenting path in $M + \OPT$ of length $2k+1, k \ge 2$. Recall our assumption that, without loss of generality, $C$ is of the form $a_0 - b_0 - a_1 - \dotsc -  a_k - b_k$, where $a_0 \in A$. Then 
		\begin{align*}
		\cost(C) =& \underbrace{\cost(\{a_0, b_0\})}_{\ge L \ \text{by Remark~\ref{rem:cost_M_augmenting_start}}} + \sum_{i=1}^{k-1} \underbrace{\cost(\{a_i, b_i\})}_{\ge L+1 \ \text{by Lemma~\ref{lem:cost_OPT_edge}}} + \underbrace{\cost(\{a_k, b_k\})}_{\ge 2L-1 \ \text{by Remark~\ref{rem:cost_M_augmenting_end_small}}} \ge
		\\ &L + (L+1)(k-1) + 2L-1=
		\\& (L+1)(k-1)+2(L+1) + (L-3)=
		\\&(L+1)(k+1)+(L-3) = 
		\\&(L+1)|OPT \cap C| + (L-3).
		\end{align*}
	
	By Remark~\ref{rem:cost_M_augmenting_start}, $b_0$ rejected a proposal from $a_0$ at some point, and $b_0$ points right. Suppose now that the above inequality is tight only. Then Corollary~\ref{cor:cost_M_augmenting_relationship} implies that, for all $i = 0, \ldots, k-1$, $b_i$ rejected a proposal from $a_i$, and $b_i$ points right. But then, Lemma~\ref{lem:cost_M_augmenting_end_nodes_relationship} implies that $\cost(\{a_{k-1}, b_{k-1}\}) \ge L+2$ holds, contradicting that the above inequality is tight. Thus we get the desired inequality $\cost(C) \ge (L+1)|\OPT\cap C| + (L-2)$.
\qed
\end{proof}

\begin{proof*}{Proof of Lemma~\ref{lem:cost_M_augmenting path}}\label{proof_lem:cost_M_augmenting path}
	By Corollary~\ref{cor:cost_not_M_augmenting path} and Remark~\ref{rem:cost_trivial_component}, for every connected component $C$ in $M + \OPT$ that is not an $M$-augmenting path, $\cost(C) \ge (L+1)|\OPT \cap C|$ holds. Also, by Lemma~\ref{lem:cost_M_augmenting}, for each connected component $C$ in $M + \OPT$ that is an $M$-augmenting path of length at least 5, $\cost(C) \ge (L+1)|\OPT \cap C| + (L-2)$ holds. Since there are at least $|\OPT| - |M|$ $M$-augmenting paths in $M + \OPT$, we obtain the desired inequality $\sum_{C \in \mathcal{C}(M + \OPT)} \cost(C) \ge (L+1)|\OPT| + (L-2)(|\OPT|-|M|)$.
\qed
\end{proof*}

\begin{proof*}{Proof of Lemma~\ref{lem:cost_M_augmenting path_edge}}
\label{proof_lem:cost_M_augmenting path_edge}
	Clearly, $(a_i, b_{i-1})$ and $(a_{i+1},b_i)$ are contained in~$G'$ since they are in~$M$. Thus $\deg(a_i) \ge 1$ and $\deg(b_i) \ge 1$. Moreover, since $(a_i, b_i)$ is included in~$G$, at least one of the following is true: $\deg(a_i) = L$; and $\deg(b_i) = L$. Hence it is sufficient to consider the following cases:
	\begin{enumerate}[label*=\Roman*.]
		\item \label{case:cost_M_augmenting path_edge_proof1}  $\deg(a_i) < L$ and $\deg(b_i) = L$.
		\item \label{case:cost_M_augmenting path_edge_proof2} $\deg(a_i) = L$.
		\begin{enumerate}[label*=\Roman*.]
			\item  \label{case:cost_M_augmenting path_edge_proof2.1} $b_i$ rejected a proposal from $a_i$.
			\begin{enumerate}[label*=\Roman*.]
				\item \label{case:cost_M_augmenting path_edge_proof2.1.1} $b_i$ has at most $L-2$ good inputs.
				\item \label{case:cost_M_augmenting path_edge_proof2.1.2} $b_i$ has $L-1$ good inputs.
			\end{enumerate}
			\item \label{case:cost_M_augmenting path_edge_proof2.2}  $b_i$ did not reject any proposal from $a_i$.
			\begin{enumerate}[label*=\Roman*.]
			\item\label{case:cost_M_augmenting path_edge_proof2.2.1} there is an edge $(a_i, b_i)$ in~$G'$.
				\item\label{case:cost_M_augmenting path_edge_proof2.2.2} there is not an edge $(a_i, b_i)$ in~$G'$.
				\begin{enumerate}[label*=\Roman*.]
					\item \label{case:cost_M_augmenting path_edge_proof2.2.2.1} $a_i$ has at least one bad output.
					\item \label{case:cost_M_augmenting path_edge_proof2.2.2.2} $a_i$ has $L-1$ good outputs.
				\end{enumerate}
			\end{enumerate}
		\end{enumerate}
	\end{enumerate}
	 
 In case~\eqref{case:cost_M_augmenting path_edge_proof1}, $a_i$ is unsuccessful. Thus $b_i$ rejected a proposal from $a_i$ as a 2-promoted man. Also, $b_i$ has a proposal from $a_{i+1}$ when the algorithm finishes, implying~\eqref{case:cost_M_augmenting path_edge2}.
 
 In case~\eqref{case:cost_M_augmenting path_edge_proof2.1.1}, $\deg(b_i) = L$ holds since $b_i$ rejected a proposal from $a_i$ at some point during the algorithm. Since $b_i$ has at most $L-2$ good inputs, $\cost(b_i) \ge \deg(b_i) - (L-2) = 2$ holds. Thus
 \[
 \cost(\{a_i, b_i\}) =  \underbrace{\cost(a_i)}_{\ge \deg(a_i) = L} + \underbrace{\cost(b_i)}_{\ge 2} \ge L+2\,,
 \]
 implying~\eqref{case:cost_M_augmenting path_edge1}.
 
 In case~\eqref{case:cost_M_augmenting path_edge_proof2.1.2}, $b_i$ is popular since she rejected a proposal at some point. Let $(a^j, b_i)$ for all $j=1,\ldots,L-1$ be good inputs to $b_i$. Then, by definition of good inputs, $a_i \ge_{b_i} a^j$ for all $j=1,\ldots,L-1$. Also, $a^j \ne a_i$ for all $j=1,\ldots,L-1$ because $(a_i, b_i) \in \OPT$, $(a_{i+1}, b_i) \in M$ and $(a^j, b_i)$ for all $j=1,\ldots,L-1$ are good inputs. Since $b_i$ rejected a proposal from $a_i$ at some point while she has proposals from $a_{i+1}$ and $a^j$ for $j=1,\ldots,L-1$ when the algorithm ends, we deduce that $a_i \le_{b_i} a_{i+1}$ and $a_i \le_{b_i} a^j$ for all $j=1,\ldots,L-1$. Since $a_i \ge_{b_i} a^j$ and $a_i \le_{b_i} a^j$, we conclude that $a_i \simeq_{b_i} a^j$ for all $j=1,\ldots,L-1$. 
 
 Since $a_i \le_{b_i} a_{i+1}$, $a_i \simeq_{b_i} a^j$, $a^j \ne a_i$ for  all $j=1,\ldots,L-1$, and ties are of size at most $L$, at least one of the following is true:
 \begin{enumerate}[label*=\roman*.]
 	\item \label{case:cost_M_augmenting path_edge_proof_extra_first_1} $a_i <_{b_i} a_{i+1}$.
 	\item \label{case:cost_M_augmenting path_edge_proof_extra_first_2}$a_i \simeq_{b_i} a_{i+1}$.
 	\begin{enumerate}[label*=\roman*.]
		\item \label{case:cost_M_augmenting path_edge_proof_extra_first_2.1}  there exist $ j',j'' =1,\ldots,L-1$, $j' \ne j''$ such that $a^{j'} = a^{j''}$.
		\item \label{case:cost_M_augmenting path_edge_proof_extra_first_2.2} there exists $ j' =1,\ldots,L-1$ such that $a^{j'} = a_{i+1}$.
	\end{enumerate}
 \end{enumerate}
 
 In case~\eqref{case:cost_M_augmenting path_edge_proof_extra_first_1}, we immediately get~\eqref{case:cost_M_augmenting path_edge2}.
 
 In case~\eqref{case:cost_M_augmenting path_edge_proof_extra_first_2.1}, by definition of good inputs, $a^{j'}$ is  either  basic or 2-promoted. If $a^{j'}$ is basic, that is in contradiction to the rejection step since $a_i \simeq_{b_i} a^{j'}$, $b_i$ rejected a proposal from $a_i$ at some point, $b_i$ holds no proposal from $a_i$ while she holds two proposals from $a^{j'}$ when the algorithm terminates. If $a^{j'}$ is 2-promoted, then $a_{i+1}$ is not basic because $a_{i+1} \simeq_{b_i} a^{j'}$, $b_{i}$ rejected $a^{j'}$ as a 1-promoted man while she holds a proposal from $a_{i+1}$ when the algorithm ends. Thus we conclude~\eqref{case:cost_M_augmenting path_edge2}.
 
 In case~\eqref{case:cost_M_augmenting path_edge_proof_extra_first_2.2}, if $a_{i+1}$ is basic, that is in contradiction to the rejection step because  $a_{i} \simeq_{b_i} a_{i+1}$, $b_i$ rejected a proposal from $a_i$ at some point, $b_i$ holds no proposal from $a_i$  while she holds two proposals from $a_{i+1}$ when the algorithm finishes. Thus $a_{i+1}$ cannot be basic, implying that $b_i$ points right. Hence we deduce~\eqref{case:cost_M_augmenting path_edge2}.
 
 In case~\eqref{case:cost_M_augmenting path_edge_proof2.2.1}, $\cost(a_i) \ge L$ holds. Also, by Remark \ref{rem:M_and_OPT_cost}, $\cost(b_i) \ge 2$ holds since $(a_i, b_i) \in \OPT$, $(a_i, b_i) \in~G'$, and $(a_{i+1}, b_i) \in M$. Thus $\cost(\{a_i, b_i\})=\cost(a_i)+\cost(b_i) \ge L+2$, implying~\eqref{case:cost_M_augmenting path_edge1}. 
 
 In case~\eqref{case:cost_M_augmenting path_edge_proof2.2.2.1}, since $(a_{i+1}, b_i) \in M$, $\cost(b_i) \ge 1$ holds. Also, because $a_i$ has at least one bad output, $\cost(a_i) \ge \deg(a_i) + 1 = L+1$ holds. Thus $\cost(\{a_i, b_i\}) = \cost(a_i) + \cost(b_i) \ge L+2$, implying~\eqref{case:cost_M_augmenting path_edge1}.
 
 In case~\eqref{case:cost_M_augmenting path_edge_proof2.2.2.2}, let $(a_i, b^j)$ for  $j = 1,\ldots,L-1$ be good outputs from $a_i$. Since $b_i$ did not reject any proposal from $a_i$ during the algorithm, $(a_i, b_{i-1}) \in M$, and $(a_i, b^j)$ for  $j = 1,\ldots,L-1$ are outputs, we deduce that $a_i$ is basic, $b_i \le_{a_i} b_{i-1}$, $b_i \le_{a_i} b^j$ for all $j = 1,\ldots,L-1$. Since $a_i$ is basic and $(a_i, b^j)$ for all $j = 1,\ldots,L-1$ are good outputs from $a_i$, we deduce that, by definition of good outputs, $b_i \ne b^j$, $b_i \ge_{a_i} b^j$ for all $j = 1,\ldots,L-1$, and hence $b_i \simeq_{a_i} b^j$ for all $j = 1,\ldots,L-1$. 
 
 Because $b_i \le_{a_i} b_{i-1}$, $b_i \simeq_{a_i} b^j$, $b_i \ne b^j$ for all  $j = 1,\ldots,L-1$, and ties are of size at most $L$, at least one of the following is true:
 \begin{enumerate}[label*=\roman*.]
 	\item \label{case:cost_M_augmenting path_edge_proof_extra_second_1} $b_i <_{a_i} b_{i-1}$.
 	\item \label{case:cost_M_augmenting path_edge_proof_extra_second_2} $b_i \simeq_{a_i} b_{i-1}$.
 	\begin{enumerate}[label*=\roman*.]
 		\item \label{case:cost_M_augmenting path_edge_proof_extra_second_2.1} there exist $ j',j'' =1,\ldots,L-1$, $j' \ne j''$  such that $b^{j'} = b^{j''}$.
 		\item \label{case:cost_M_augmenting path_edge_proof_extra_second_2.2}  there exists $ j' =1,\ldots,L-1$ such that $b^{j'} = b_{i-1}$.
 	\end{enumerate}
 \end{enumerate}
 
In case~\eqref{case:cost_M_augmenting path_edge_proof_extra_second_1}, we immediately get~\eqref{case:cost_M_augmenting path_edge3}.

In cases~\eqref{case:cost_M_augmenting path_edge_proof_extra_second_2.1} and~\eqref{case:cost_M_augmenting path_edge_proof_extra_second_2.2}, by definition of good outputs, $b^{j'}$ is popular and so she rejected a proposal at some point. On the other hand, $b_i \simeq_{a_i} b^{j'}$, $b^{j'}$ holds at least two proposals from $a_i$ when the algorithm finishes, $b_i$ did not reject $a_i$ during the algorithm, and $b_i$ does not hold any proposal from $a_i$ when the algorithm terminates, a contradiction to the forward step for $b^{j'}$.
\qed
\end{proof*}

The following remark is used to simplify the proof of Lemma~\ref{lem:cost_M_augmenting_end_nodes_relationship} below.

\begin{remark}\label{rem:cost_M_augmenting_end_fourth_end_node}
	Let $a_0-b_0-a_1- \dotsc -a_k-b_k$ be an $M$-augmenting path in $M + \OPT$ of length $2k +1$, $k \ge 2$, where $a_0 \in A$. Then $\deg(a_{k-1}) = L.$
\end{remark}
\begin{proof}
	Suppose for a contradiction that $\deg(a_{k-1}) < L$, and so $a_{k-1}$ is unsuccessful. Thus $b_{k-1}$ rejected $a_{k-1}$ as a 2-promoted man, and so $b_{k-1}$ is popular. Since $b_k$ is unmatched in~$M$, $\deg(b_k) < L$ holds. So $b_k$ is unsuccessful, and thus $a_k$ is basic.  Since $b_{k-1}$ rejected a proposal from $a_{k-1}$ at some point, and $b_{k-1}$ has a proposal from $a_k$ when the algorithm ends, we deduce that $a_{k-1} \le_{b_{k-1}} a_k$.
	
	First, if $a_{k-1}\simeq_{b_{k-1}} a_k$ holds, we deduce that $a_k$ is 2-promoted, contradicting the fact that $a_k$ is basic. Thus $a_{k-1} <_{b_{k-1}} a_k$. Since $b_k$ did not reject any proposal from $a_k$ during the algorithm, and $b_{k-1}$ holds a proposal from $a_k$ when the algorithm terminates, $b_{k-1} \ge_{a_{k}} b_{k}$ holds. If $b_{k-1} >_{a_{k}} b_{k}$ holds, then $(a_k, b_{k-1})$ is a blocking pair for OPT, contradicting the stability of OPT. We conclude that $b_{k-1} \simeq_{a_{k}} b_{k}$. But then, since $b_{k-1}$ has a proposal from $a_k$ when the algorithm finishes, $b_k$ is unsuccessful, and $b_{k-1} \simeq_{a_k} b_k$, Remark~\ref{rem:bounce_step} implies that $b_{k-1}$ is unpopular, a contradiction.
\qed
\end{proof}

\begin{proof*}{Proof of Lemma~\ref{lem:cost_M_augmenting_end_nodes_relationship}}
	Observe that $(a_k, b_{k-1}) \in M$, and thus $(a_k, b_{k-1}) \in G'$. Therefore, one of the following cases is true:
	\begin{enumerate}[label*=\Roman*.]
		\item \label{case:cost_M_augmenting_end_nodes_relationship1} there is at least one edge $(a_{k-1}, b_{k-1})$ in~$G'$. 
		\item \label{case:cost_M_augmenting_end_nodes_relationship2} there is no edge $(a_{k-1}, b_{k-1})$ in~$G'$.
		\begin{enumerate}[label*=\Roman*.]
			\item \label{case:cost_M_augmenting_end_nodes_relationship2.2.1} there are at least two parallel edges $(a_k, b_{k-1})$ in~$G'$.
			\item \label{case:cost_M_augmenting_end_nodes_relationship2.2.2} there are exactly $L-1$ edges, $(a^j, b_{k-1})$ for $j=1,\ldots,L-1$ in~$G'$, and $a^j \ne a_{k-1}$, $a^j \ne a_{k}$ for $j=1,\ldots,L-1$.
		\end{enumerate}
	\end{enumerate}
	
	In case~\eqref{case:cost_M_augmenting_end_nodes_relationship1}, Remark~\ref{rem:cost_M_augmenting_end_fourth_end_node} implies $\cost(a_{k-1}) \ge \deg(a_{k-1}) = L$, and Remark~\ref{rem:M_and_OPT_cost} implies $\cost(b_{k-1}) \ge 2$ since $(a_{k-1}, b_{k-1}) \in \OPT$, $(a_{k-1}, b_{k-1}) \in~G'$, and $(a_k, b_{k-1}) \in M$. Thus,
	\[
	\cost(\{a_{k-1}, b_{k-1}\}) = \underbrace{\cost(a_{k-1})}_{\ge L} + \underbrace{\cost(b_{k-1})}_{\ge 2} \ge L+2\,,
	\] as desired.

	In case~\eqref{case:cost_M_augmenting_end_nodes_relationship2.2.1}, since $b_{k-1}$ rejected a proposal from $a_{k-1}$ at some point, and $b_{k-1}$ holds at least two proposals from $a_k$ when the algorithm terminates, we deduce that $b_{k-1}$ is popular and $a_k \ge_{b_{k-1}} a_{k-1}$ holds. Also, since $b_k$ is unmatched in~$M$, $\deg(b_k) < L$ holds. Thus $b_k$ is unsuccessful and $a_k$ is basic. If $a_k \simeq_{b_{k-1}} a_{k-1}$, then Remark~\ref{rem:rejection_step} implies that there is an edge $(a_{k-1}, b_{k-1})$ in~$G'$, a contradiction. Thus $a_k >_{b_{k-1}} a_{k-1}$ holds.

We show that $a_k >_{b_{k-1}} a_{k-1}$ leads to a contradiction. Since $b_k$ is unsuccessful and there is an edge $(a_k, b_{k-1})$ in~$G'$, $b_{k-1} \ge_{a_k} b_{k}$ holds. If $b_{k-1} >_{a_k} b_k$, then $(a_k, b_{k-1})$ is a blocking pair for OPT, contradicting the stability of OPT. Thus $b_{k-1} \simeq_{a_k} b_k$ holds. But then, since $b_{k-1}$ holds a proposal from $a_k$ when the algorithm ends, $b_k$ is unsuccessful, and $b_{k-1} \simeq_{a_k} b_k$, Remark~\ref{rem:bounce_step} implies that $b_{k-1}$ is unpopular, a contradiction.
	
	In case~\eqref{case:cost_M_augmenting_end_nodes_relationship2.2.2}, since  $b_{k-1}$ holds proposals from $a_k$ and $a^j$ for all $j=1,\ldots,L-1$, and $b_{k-1}$ rejected a proposal from $a_{k-1}$ at some point, we deduce that $b_{k-1}$ is popular, $a_k \ge_{b_{k-1}} a_{k-1}$, $a^j \ge_{b_{k-1}} a_{k-1}$ for all $j=1,\ldots,L-1$. Since $b_k$ is unmatched in~$M$, $\deg(b_k) < L$ holds and therefore $b_k$ is unsuccessful.
	
	Analogously to the proof of case~\eqref{case:cost_M_augmenting_end_nodes_relationship2.2.1}, it can be shown that $a_k >_{b_{k-1}} a_{k-1}$ leads to a contradiction. Thus $a_k \simeq_{b_{k-1}} a_{k-1}$, and $a^j \ge_{b_{k-1}} a_{k-1}$ for all $j=1,\ldots,L-1$ hold. Since $a^j \ne a_{k-1}$, $a^j \ne a_{k}$ for all $j=1,\ldots,L-1$, and ties are of size at most $L$, at least one of the following is true:
	\begin{enumerate}[label*=\roman*.]
		\item \label{case:cost_M_augmenting_path proof_extra_second1} there exists $j' =1,\ldots,L-1$ such that $a^{j'} >_{b_{k-1}} a_{k-1}$.
		\item \label{case:cost_M_augmenting_end_nodes_relationship_extra_second2} $a^j \simeq_{b_{k-1}} a_{k-1}$ for all $j =1,\ldots,L-1$.
		\begin{enumerate}[label*=\roman*.]
			\item \label{case:cost_M_augmenting_end_nodes_relationship_extra_second2.1} there exist $ j',j'' =1,\ldots,L-1$, $j' \ne j''$ such that $a^{j'} = a^{j''}$.
		\end{enumerate}
	\end{enumerate}
	
	In case~\eqref{case:cost_M_augmenting_path proof_extra_second1}, $(a^{j'}, b_{k-1})$ is a bad input to $b_{k-1}$ by definition. Thus, by Remark~\ref{rem:bad_inputs_cost1}, $\cost(b_{k-1}) \ge 2$ holds. Since $\cost(a_{k-1}) \ge \deg(a_{k-1}) = L$ holds by Remark~\ref{rem:cost_M_augmenting_end_fourth_end_node}, we obtain the desired inequality
	\[
	\cost(\{a_{k-1}, b_{k-1}\}) =  \underbrace{\cost(a_{k-1})}_{\ge L} + \underbrace{\cost(b_{k-1})}_{\ge 2} \ge L+2\,.
	\]
	
	In case~\eqref{case:cost_M_augmenting_end_nodes_relationship_extra_second2.1}, recall that there is no edge $(a_{k-1}, b_{k-1})$ in~$G'$. Since $b_{k-1}$ rejected a proposal from $a_{k-1}$ during the algorithm, $a^{j'} \simeq_{b_{k-1}} a_{k-1}$, and $(a_{k-1}, b_{k-1}) \notin G'$, we deduce from Remark~\ref{rem:rejection_step} that $a^{j'}$ is not basic. If $a^{j'}$ is 1-promoted, then $(a^{j'}, b_{k-1})$ and $(a^{j''}, b_{k-1})$ are bad inputs to $b_{k-1}$ by definition. Thus, by Remark~\ref{rem:bad_inputs_cost1},  $\cost(b_{k-1}) \ge 3$ holds. Since $\cost(a_{k-1}) \ge \text{deg} (a_{k-1}) = L$ holds by Remark~\ref{rem:cost_M_augmenting_end_fourth_end_node}, we get the desired inequality
	\[
	\cost(\{a_{k-1}, b_{k-1}\}) =  \underbrace{\cost(a_{k-1})}_{\ge L} + \underbrace{\cost(b_{k-1})}_{\ge 3} \ge L+2\,.
	\]
	If $a^{j'}$ is 2-promoted, then $b_{k-1}$ rejected $a^{j'}$ as a 1-promoted man. On the other hand, $a_k \simeq_{b_{k-1}} a^{j'}$, $a_k$ is basic, and $b_{k-1}$ holds a proposal from $a_k$ when the algorithm ends, a contradiction to the rejection step.	
\qed
\end{proof*}

\subsection{Tightness of the analysis}
The following example shows that the bound in Theorem~\ref{thm:main result} is tight. 

\begin{figure}
    \centering
    \caption{An instance with ties of size at most $L$, $L \ge 2$ for which the algorithm outputs a stable matching $M$ with $|\OPT| / |M| = (3L-2)/(2L-1)$}
    \begin{displaymath}
    \arraycolsep=0.5cm
    \begin{array}{ll}
    
   		\text{Men's preferences} & \text{Women's preferences} \\ [2mm]
   		 
        a_0 : (b_0 \enspace b_1^{\gamma} \, \ldots \, b_{L-1}^{\gamma}) & b_0 : (a_0 \enspace a_1^{\beta} \, \ldots \, a_{L-1}^{\beta}) \\ [3mm]
        
        a_1^{\alpha} : (b_1^{\alpha} \enspace b_1^{\gamma} \, \ldots \, b_{L-1}^{\gamma}) & b_1^{\alpha} : a_1^{\alpha} \enspace a_1^{\beta} \, \ldots \, a_{L-1}^{\beta} \\ 
        
        \ \vdots & \ \vdots\\ 
        
        a_{L-1}^{\alpha} : (b_{L-1}^{\alpha} \enspace b_1^{\gamma} \, \ldots \, b_{L-1}^{\gamma}) & b_{L-1}^{\alpha} : a_{L-1}^{\alpha} \enspace a_1^{\beta} \, \ldots \, a_{L-1}^{\beta} \\ [3mm]
        
        a_1^{\beta} : (b_0 \enspace b_1^{\alpha} \, \ldots \, b_{L-1}^{\alpha}) \enspace b_1^{\beta} & b_1^{\beta} : a_1^{\beta} \\ 
        
        \ \vdots & \ \vdots\\ 
        
        a_{L-1}^{\beta} : (b_0 \enspace b_1^{\alpha} \, \ldots \, b_{L-1}^{\alpha}) \enspace b_{L-1}^{\beta} & b_{L-1}^{\beta} : a_{L-1}^{\beta} \\ [3mm]
        
        a_1^{\gamma} : b_1^{\gamma} & b_1^{\gamma} : (a_0 \enspace a_1^{\alpha} \, \ldots \, a_{L-1}^{\alpha}) \enspace a_1^{\gamma} \\ 
        
        \ \vdots & \ \vdots\\ 
        
        a_{L-1}^{\gamma} : b_{L-1}^{\gamma} & b_{L-1}^{\gamma} : (a_0 \enspace a_1^{\alpha} \, \ldots \, a_{L-1}^{\alpha}) \enspace a_{L-1}^{\gamma} \\ 

    \end{array}
    \end{displaymath}
    \label{Example 1}
\end{figure}

\begin{example}
In Figure~\ref{Example 1}, the preference list of each individual is ordered from a most preferred person to a least preferred one, where individuals within parentheses are tied. For example, $a_1^{\beta}$ is indifferent between all the women in his preference list except $b_1^{\beta}$, who is less preferred than the others.

It is straightforward to check that there exists a unique maximum-cardinality stable matching, namely $\OPT = \{(a_0, b_0)\} \cup \{(a_i^{j}, b_i^{j}) \mid i = 1,\ldots,L-1, \ j = \alpha,\beta, \gamma \}$. We show that there exists an execution of the algorithm which outputs the matching $M = \{(a_0, b_0)\} \cup \{(a_i^{\alpha}, b_i^{\gamma}) \mid i = 1,\ldots,L-1 \} \cup \{(a_i^{\beta}, b_i^{\alpha}) \mid i = 1,\ldots,L-1 \}$, leading to the ratio $|\OPT| / |M| = (3L-2)/(2L-1)$. 
\end{example}

\begin{proof}
The following is an execution of the algorithm which leads either to the matching $M$ or a matching with the size of $M$.
\begin{itemize}
\item $a_0$ makes one proposal to every woman in his list; the women accept.
\item $a_i^{\alpha}$ for all $i = 1,\ldots,L-1$ makes one proposal to every woman in his list; the women accept.
\item $a_i^{\beta}$ for all $i = 1,\ldots,L-1$ makes one proposal to every woman except the last one in his list; the women accept.
\item $a_i^{\gamma}$ starts to propose $b_i^{\gamma}$ for all $i = 1,\ldots,L-1$, but each time $a_i^{\gamma}$ makes a proposal, the proposal is rejected; $a_i^{\gamma}$ gives up.
	
\end{itemize}
\qed
\end{proof}

\newpage
\bibliographystyle{splncs04}
\bibliography{bibliography}

\end{document}